\documentclass[letterpaper, 10 pt, conference]{ieeeconf}  
\IEEEoverridecommandlockouts                  
\let\labelindent\relax
\usepackage{dsfont}
\usepackage{amsfonts}
\usepackage{amsmath}
\usepackage{amssymb}
\usepackage{mathtools}
\usepackage{pifont}
\DeclareMathOperator*{\argmax}{arg\,max}

\usepackage[shortlabels]{enumitem}
\usepackage[normalem]{ulem}
\usepackage{tablefootnote}
\usepackage{cite}
\usepackage{booktabs}
\usepackage{multicol}
\usepackage{multirow}
\usepackage{tabularx}
\usepackage{diagbox}
\usepackage{graphicx}
\usepackage{subcaption}

\usepackage{array}
\usepackage{fancyhdr}
\usepackage{hyperref}
 \hypersetup{
     colorlinks=true,
     linkcolor=blue,
     filecolor=blue,
     citecolor = black,      
     urlcolor=cyan,
     }
\usepackage{cleveref}
\usepackage{xcolor}
\usepackage{tikz}
\usepackage{verbatim}
\usepackage{listings}
\definecolor{mygreen}{RGB}{28,172,0} 
\definecolor{mylilas}{RGB}{170,55,241}

\usepackage{matlab-prettifier}
\usepackage{color}

\usepackage{algorithm}
\usepackage[noend]{algpseudocode}
\makeatletter
\def\BState{\State\hskip-\ALG@thistlm}
\makeatother
\newcommand{\supp}{\operatorname{supp}}

\newcommand{\E}{\mathbb{E}}

\newcommand{\R}{\mathbb{R}}

\newtheorem{theorem}{Theorem}
\newtheorem{corollary}{Corollary}

\newtheorem{proposition}{Proposition}

\newtheorem{assumption}{Assumption}
\newtheorem{remark}[theorem]{Remark}

\newcommand{\Rmnum}[1]{\expandafter\@slowromancap\romannumeral #1@}

\title{\LARGE \bf Designing Policies for Truth: Combating Misinformation with Transparency and Information Design   
}

\author{Ya-Ting Yang, Tao Li, and Quanyan Zhu
\thanks{The Authors are with the Department of Electrical and Computer Engineering, New York University, Brooklyn, NY, 11201, USA; E-mail: {\tt\small \{yy4348, tl2636, qz494\}@nyu.edu}. YY and TL have contributed equally. Correspondence should be addressed to TL.}%
}

\begin{document}

\maketitle
\thispagestyle{empty}
\pagestyle{empty}

\begin{abstract}
Misinformation has become a growing issue on online social platforms (OSPs), especially during elections or pandemics. To combat this, OSPs have implemented various policies, such as tagging, to notify users about potentially misleading information. However, these policies are often transparent and therefore susceptible to being exploited by content creators, who may not be willing to invest effort into producing authentic content, causing the viral spread of misinformation. Instead of mitigating the reach of existing misinformation, this work focuses on a solution of prevention, aiming to stop the spread of misinformation before it has a chance to gain momentum. We propose a Bayesian persuaded branching process ($\operatorname{BP}^2$) to model the strategic interactions among the OSP, the content creator, and the user. The misinformation spread on OSP is modeled by a multi-type branching process, where users' positive and negative comments influence the misinformation spreading. Using a Lagrangian induced by Bayesian plausibility, we characterize the OSP's optimal policy under the perfect Bayesian equilibrium. The convexity of the Lagrangian implies that the OSP's optimal policy is simply the fully informative tagging policy: revealing the content's accuracy to the user. Such a tagging policy solicits the best effort from the content creator in reducing misinformation, even though the OSP exerts no direct control over the content creator. We corroborate our findings using numerical simulations. 
\end{abstract}

\section{Introduction}
Misinformation has become a growing concern on online social platforms (OSP), as false information can spread rapidly and have significant consequences \cite{Zhao2020-zu}. For instance, false stories about candidates were shared widely through OSPs during the 2016 US presidential election; misinformation about the virus, mask-wearing policies, and vaccine concerns spread through social networks during the COVID-19 pandemic. To address this issue, OSPs have implemented policies such as labeling, tagging, or notifying to alert users to potentially false or misleading information \cite{twitter, fb_transparency}. Previous studies have shown that these policies effectively (to some extent) curb the spread of misinformation \cite{ Platform_intervention}.

However, as mandated by related regulations or ethical standards \cite{Crilley2019-xz, zhu23dce}, these policies are often transparent, meaning that they are publicly announced to the content creators (creators) and users. Aware of the OSP's policies, creators take advantage of this transparency by spending the least possible effort to make their posts as trustworthy as possible so as to pass the screening. On the other hand, OSPs are constantly upgrading their policies to combat such tactics. As the two parties fall into an endless arms race caused by the conflict of interests, it is natural to ask: \emph{ Does the two reach an equilibrium in the end? } Does the transparency requirement give content creators the upper hand? 

\begin{figure}
    \centering
    \includegraphics[width=3.3in]{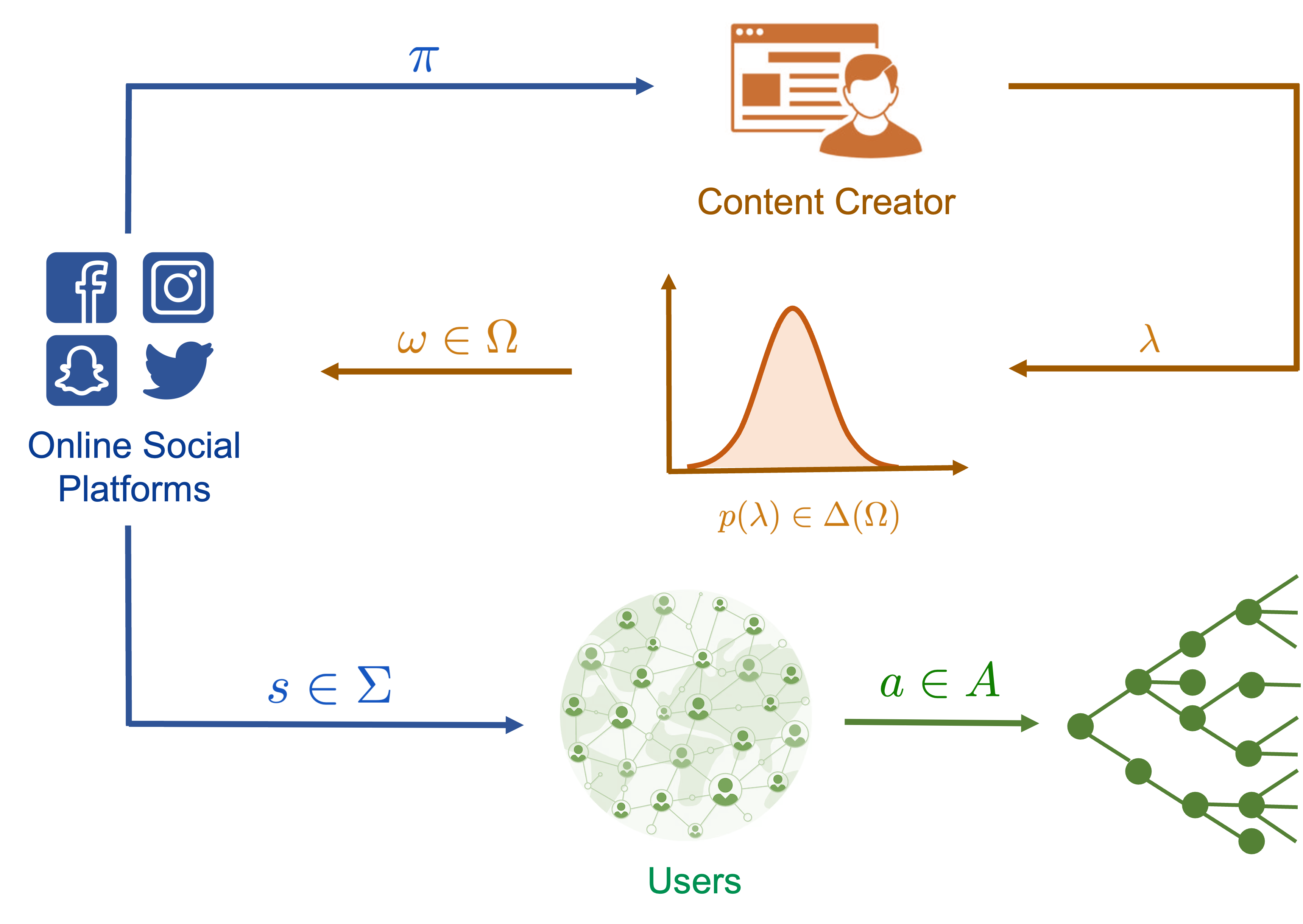}
    \caption{The Bayesian Persuaded Branching Process. The OSP first commits to an information structure $\pi$, followed by a private effort $\lambda$ exerted by the content creator, influencing the distribution of true/fake posts $\omega$. The users offer positive/negative comments to the post after observing the realized tag/label $s$, and then forward it to others. }
    \label{fig:proposed_framework}\vspace{-5mm}
\end{figure}

To address these questions, we propose a persuasion game that captures the interactions among the OSP, the content creator, and the user as illustrated in Fig. \ref{fig:proposed_framework}. The OSP designs a tagging policy whose realized tags indicate the content accuracy of an arbitrary post. Such a policy does not directly control any decisions or utilities but influences others' behaviors through information provision. Hence, this tagging policy is referred to as the  information structure \cite{tao_info}, and the OSP's problem is termed information design. Fully aware of this policy, the content creator exerts a private effort in creating the content, with the assumption that the more effort exerted, the more accurate the content is. Finally, the user observes the tagging policy and the realized tags, then decides on their views and comments. The OSP aims to persuade the user not to facilitate the misinformation circulation and incentivize the content creator to spend the highest possible effort (i.e., not to create misinformation).

The proposed model differs from the seminal Bayesian persuasion model \cite{kamenica11BP} in that the user cannot directly observe the prior distribution. As a result, the receiver must form a conjecture about the content creator's behavior to update their beliefs. This conjecture must be consistent with the agent's equilibrium behavior, which leads to the concept of perfect Bayesian equilibrium (PBE) as the natural solution concept for our game. In addition, the user's action (e.g., comment and share) might influence the trends on the social platform \cite{fighting_comment}. Hence, we use the branching process (BP) to capture misinformation spreading \cite{BP_cascade_twitter}, which affects the content creator's reputation and the OSP's payoff. 


This research demonstrates that a simple information structure can be a powerful tool in combatting misinformation spread. By adopting a fully informative policy, such as using tagging to indicate content accuracy, content creators are incentivized to produce trustworthy material. Although the OSP may not have direct control over content creators, it can nudge user perceptions through the information structure. The collective behaviors of users, under these perceptions, determine the content creators' reputations, effectively making users the OSP's proxy in terms of incentive provision. Our contributions can be summarized as follows.
\begin{itemize}
    \item We propose a three-player persuasion game to capture the interactions induced by the conflict of interests between the OSP, the content creator, and the user, where the multi-type branching process is utilized to model the spread of misinformation content;
    \item We develop a Lagrangian approach to identify players' strategies under perfect Bayesian equilibrium, which is known to be challenging to solve \cite{dughmi19hardness,tao22bp, tao23pot}. Through Bayesian plausibility, we transform the equilibrium problem into the posterior belief space and develop an equality-constrained nonlinear programming associated with the equilibria, facilitating the study of the optimal information structure. 
 
\end{itemize}

\noindent\textbf{Related works} Existing research on misinformation typically focuses on a finite set of players connected by a graph, with the reliability of articles, news, and content drawn from a known distribution \cite{model_online_mis, enga_mis}. For instance, \cite{model_online_mis} has proposed a model for online sharing behavior of fully Bayesian users under potential misinformation and studies the impact of network structure, demonstrating that social media platforms that aim to maximize engagement might help propagate misinformation. \cite{enga_mis} has investigated how the platform should design a signaling mechanism to influence users' engagement for maximizing engagement or minimizing misinformation purposes. 

In contrast, our approach considers the population-wide effects of misinformation circulation. Specifically, we analyze the proportion of individuals receiving negative comments among all receivers using branching processes, which is shown to match well to statistical characteristics of information cascades \cite{BP_cascade_twitter}, and has been utilized in studying the determinants behind misinformation spreading \cite{spreading_mis}. Rather than analyzing misinformation circulation\cite{BPBP_tagging} through branching processes, we aim to prevent it from being created in the first place. To study this preemptive solution, we introduce a third player, i.e., the content creator, into the Bayesian persuasion framework, where the OSP incentivizes the content creator to produce accurate content. The ultimate goal is to curb misinformation spread by promoting the creation of truthful and reliable content.

\section{Online Misinformation Circulation: Modeling and Information Design}
This section introduces a three-player persuasion game that models the interactions among the online social platform, the content creators, and the users. Naturally, misinformation circulation on OSP involves a population of content creators and users. To simplify the exposition, we consider a representative creator and a population of users with identical utilities. We pick a representative user, referred to as  ``the user,'' when discussing strategic reasoning in the persuasion game, as the population all share the same interest. In contrast, we consider ``users'' when treating population-level misinformation dissemination using branching processes.       

\subsection{The Bayesian Persuaded Branching Processes Model}
In the persuasion game, the OSP designs a tagging policy about a state variable that reflects the accuracy of the content of the post. Fully aware of this tagging policy, the content creator then exerts a private effort, which is unobservable to both the OSP and the user, to determine the distribution of the accuracy of the content. Less effort leads to more misinformation prevailing over social media. Finally, the user takes action by commenting on the post and sharing after knowing the tagging policy and observing the tag realization. Note that the state variable remains hidden from the user throughout the game, as individuals lack the necessary resources to verify the authenticity.   

The action taken by the user results in a \emph{trend} (negative or positive about the post) in social media. To understand this notion, we consider a multi-type branching process (introduced later in \Cref{sec:branching}). Denote by $X(t)$ the number of users who have just received the post with a negative comment at time $t$ ($x$-type user). Similarly, $Y(t)$ denotes the number of users who have received a positive comment ($y$-type user). After reading the received post, users forward it to some of their followers/friends with their own (either negative or positive) comments, producing ``offsprings'' (the new $x/y$-type users). The trend is measured through the proportion of negative comments over all the comments: $\eta(t)={X(t)}/(X(t)+Y(t))$.  

For the rest of the paper, we refer to the OSP, the content creator, and the user as the sender, the agent, and the receiver, respectively, following the custom in persuasion literature \cite{kamenica11BP}.  To summarize the discussion above, the persuasion game is given by the tuple $\left\langle \Omega, \Sigma, \Lambda, p, \mathcal{A}, \eta^*, u_S, u_A, u_R,  \right\rangle$, where 
\begin{enumerate}[i)]
    \item $\Omega$ is the state space endowed with Borel algebra, and $\omega\in \Omega$ reflects how accurate the content of the post is;
    \item $\Sigma$ is the signal space of the sender (Borel algebra), and $s\in \Sigma$ denotes the tag associated with the post;
    \item $\Lambda$ is the action set of the agent, and each $\lambda\in \Lambda$ represents how much effort the agent exerts in producing trustworthy content;
    \item $p: \Lambda\rightarrow \Delta(\Omega)$ is the control function of the agent, whose effort $\lambda$ is turned into the state distribution $p(\cdot|\lambda)$ over the accuracy of the content $\Omega$;
    \item $\mathcal{A}$ is the action set of the receiver, which is a continuum $[0, 1]$, and $a \in \mathcal{A}$ denotes the probability of offering a positive comment;
    \item $\eta^*$ is the proportion of negative comment $\eta(t)$ as $t \rightarrow \infty$ obtained from the stabilized branching processes, which is related to the reputation of the agent and the impact of misinformation spreading; 
    \item $u_S: \Omega\times \mathcal{A}\rightarrow \mathbb{R}$, $u_A: \mathcal{A}\times \Lambda\rightarrow \mathbb{R}$, $u_R:\Omega\times \mathcal{A}\rightarrow\mathbb{R}$ are utility functions of the sender, the agent, and the receiver, respectively. The definitions of these utilities are as follows.
\end{enumerate} 
\paragraph{The Receiver's Utility}
To minimize the mismatch between the comment and the truth, the receiver's utility is $ u_R(\omega, a)=-(a-\omega)^2$. Suppose that the receiver believes that the state variable is subject to $\mu\in \Delta(\Omega)$, its best response under this belief is 
\begin{equation}
    a^*(\mu)=\argmax_{a\in [0,1]}\E_{\omega\sim \mu}-(a-\omega)^2=\E_{\mu}[\omega].
\label{eq:a-star}
\end{equation}  

\paragraph{The Agent's Utility}
The agent is concerned with the effort and its reputation measured through $\eta^*$ (the proportion of negative comments on its post). Denote by $c(\lambda)$ the cost induced by the effort $\lambda$; and by $r_A(a)=1-\eta^*(a)$ the agent's reputation when the receiver responds with $a$. The agent's utility is given by   
\begin{equation}
    u_A(a,\lambda)=r_A(a)-c(\lambda),
\label{eq:u_a}
\end{equation}

\paragraph{The Sender's Utility}
The sender's goal is to mitigate the influence of misinformation: if the content is misleading, the sender prefers fewer positive comments. Define 
\begin{equation}
    u_S(\omega, a)=-(1-\omega)(1-\eta^*(a))+\omega(1-\eta^*(a)).
\end{equation} Given that $1-\eta^*(a)$ represents the proportion of positive comments and $1-\omega$ indicates the level of inaccuracy of the content, the former term suggests that the sender's utility decreases when inaccurate content receives positive comments. Similarly, the latter term implies that the sender' aims to  promote accurate content.

The game unfolds in three stages. 1) In the first stage, the sender designs and commits to a tagging policy (also termed signaling) $\pi: \Omega\rightarrow \Delta(\Sigma)$, specifying a conditional distribution $\pi(\cdot|\omega)$ over the possible tags when the authenticity of a post $\omega$ is revealed. 2) Second, observing the tagging policy $\pi$, the agent chooses an effort $\lambda$ that leads to a distribution over the accuracy of the post $p(\cdot|\lambda)\in \Delta(\Omega)$. 3) Finally, when encountering an arbitrary post $\omega$ from $p(\cdot|\lambda)$, the receiver receives a tag from the tagging policy and subsequently determines their view. Note that the tagging policy $\pi$ is transparent, whereas the agent's $\lambda$ is hidden from the user.

\subsection{Perfect Bayesian Equilibrium}\label{sec:PBE}
It is worth noting that what distinguishes the introduced model from the classical Bayesian persuasion \cite{kamenica11BP} is that the receiver now does not explicitly acquire the prior distribution $p(\lambda)$, as $\lambda$ is unobservable. Hence, when the receiver acts, they must resort to a conjecture on the agent's action to update the posterior beliefs $\mu$. This conjecture must be consistent with the agent's equilibrium choice, which naturally leads to the perfect Bayesian equilibrium (PBE). The formal definition and associated details are presented in the following.

A Perfect Bayesian Equilibrium (PBE) of the proposed persuasion game consists of a tagging policy $\pi$, the agent's effort $\lambda$, and a belief system $\{\mu_s, s\in \Sigma\}$\footnote{A belief system is a collection of posterior beliefs $\mu_s$, and $\mu_s$ denotes the belief when the receiver receives signal $s$. }, which satisfies the following properties:
\begin{enumerate}[i)]
    \item given a tagging policy $\pi$ (sender) and a belief system $\{\mu_s, s\in \Sigma\}$ (receiver), the agent's effort $\lambda$ maximizes the expected utility, i.e.,
    \begin{align}
        &\lambda =\argmax \sum_{\omega}p(w|\lambda)[\sum_{s}\pi(s|\omega)u_A(\mu_s, \lambda)]\label{eq:lambda-argmax}\\
        & u_A(\mu_s, \lambda):= r_A(a^*(\mu_s), \lambda)\nonumber
    \end{align}
    \item the receiver's belief is consistent with the agent's effort $\lambda$ and the tagging policy $\pi$, i.e.,
    \begin{align}
        &\mu_s=\frac{\pi(s|\cdot)\odot p(\cdot|\lambda)}{\langle \pi(s|\cdot) ,p(\cdot|\lambda) \rangle}, \label{eq:consistency}\\
        & \pi(s|\cdot)= [\pi(s|\omega_1),\ldots, \pi(s|\omega_N)]\in \mathbb{R}^{|\Omega|}\\
        & p(\cdot|\lambda)= [p(\omega_1|\lambda), \ldots, p(\omega_N|\lambda)]\in \mathbb{R}^{|\Omega|}
    \end{align} where $\odot$ denotes the point-wise product.
    \item the tagging policy maximizes the sender's expected utility, i.e.,
    \begin{align}
        \pi\in\argmax \sum_{\omega}p(\omega|\lambda)\sum_{s}\pi(s|\omega)u_S(a^*(\mu_s), \omega ).
    \end{align}
\end{enumerate}

\subsection{Binary-State State Model}
For simplicity, we focus on a binary case study. Under such a circumstance, the state space only consists of two elements $\Omega=\{0,1\}$ with $0$ indicating the content contains misinformation while $1$ represents the content is accurate. Hence, the signal space is also assumed to be binary: $\Sigma=\{0,1\}$, where $0$ and $1$ denote the ``fake'' and ``real'' tags, respectively\footnote{In general, a sufficient signal space needs to be of $(|\Omega|+1)$-cardinality \cite{dughmi19hardness}. Yet, as we later show in \Cref{prop:optimal-is}, the binary signal space suffices.}. Since the state space is binary, the corresponding prior distribution of the accuracy of the content lives in the simplex spanned by $p_0=[1,0]$ and $p_1=[0,1]$. Hence, we assume that the effort spent by the agent $\lambda$ is a scalar from $[0,1]$, and the resulting prior distribution is the convex combination of $p_0$ and $p_1$: $p(\lambda)=(1-\lambda)p_0+\lambda p_1$. 

As the state space is finite, the players' strategies are finite-dimensional vectors, and hence, we can ``vectorize'' our analysis so that convex analysis tools can be utilized. We impose the following customary assumption \cite{kamenica11BP, boleslavsky18moral-hazard} to ensure that the agent's equilibrium problem is well-behaved. 
\begin{assumption}
    For the agent's utility given by (\ref{eq:u_a}), we assume that $r_A(\cdot)$ is non-negative and bounded, and $c(\cdot)\in C^2$ is strictly increasing and convex. In addition, $c(0)=\nabla c(0)=0$, and $\nabla c(1)>1$.
\label{assump:cost}
\end{assumption}

Let $v_A(\mu)=r_A(a^*(\mu))$ denote the agent's payoff under the receiver's belief $\mu$. Moreover, let $\Bar{v}_A(\omega|\pi):=\sum_{s}\pi(s|\omega)v_A(\mu_s)$ denote the agent's expected payoff conditional on the generated state $\omega$ under the tagging policy $\pi$, and let $\Vec{v}_A(\pi)$ be the corresponding vector: $\Vec{v}_A(\pi)=[\Bar{v}_A(0|\pi),\Bar{v}_A(1|\pi)]$. Similarly, we have the following notations for the sender. Given the receiver's belief $\mu$, the sender's expected payoff is denoted by $v_S(\mu):=\E_{\omega\sim \mu}[u_S(a^*(\mu),\omega)]$. Then, let $\Bar{v}_S(\omega|\pi)=\sum_{s}\pi(s|\omega)v_S(\mu_s)$, and therefore $\Vec{v}_S(\pi):=[\Bar{v}_S(0|\pi), \Bar{v}_S(1|\pi)]$.  

To characterize the PBE in the proposed model, we use backward induction, i.e., first analyzing the optimal actions of the receiver, then the agent, and finally the sender. 
To begin with, the receiver's best response (comment) under the belief $\mu_s$ is given by \eqref{eq:a-star}. The best-response $a^*(\mu_s)$ then affects the spread of misinformation in social media through branching processes presented in  \Cref{sec:branching}.  

\section{Content Spreading Through Branching Process} 
This section treats the spread of misinformation through branching processes. Specifically, we focus on the evolution of the trend $\eta(t)$, the proportion of negative comments, as the receiver forwards the post to others. One key finding is that the evolutionary dynamics of $\eta(t)$ under the branching process stabilizes in the limit, and the receiver's belief completely determines the stationary point $\eta^*$. 

\subsection{Multi-type Branching Processes}\label{sec:branching}
Suppose that the number of the receiver's friend $N$ is independent and identically distributed with expectation $\mathbb{E}[N]=m_{N}$ and is finite. The receiver shares the post with $Bin(N, q)$ friends, where $q \in [0, 1]$ represents the impact or attractiveness of the post (assumed to be constant). Hence, the number of ``offspring'' (friends receiving the sharing) of the receiver, denoted by $\xi$, is subject to a binomial distribution: $\xi \sim Bin(N, q)$ with $\mathbb{E}[\xi]=m_{N}\cdot q:= m$.

Denote by $\tau_n$ the time when $n^{th}$ individual ``wakes up'', meaning that such individual becomes active on an OSP and is ready to share the post. Let $X_n = X(\tau_n^{+})$, $Y_n = Y(\tau_n^{+})$, $Z_n = X_n + Y_n$, and $\xi_n \overset{i.i.d.}{\sim}Bin(N, q)$. If the $x$-type receiver (who receives negative comments) wakes up, then 
\begin{equation}
  \begin{aligned}
    X_{n+1} &= X_{n}-1 + \textbf{1}_{x} \xi_n,\\
    Y_{n+1} & = Y_{n} + \textbf{1}_{y} \xi_n,
\end{aligned}  \label{eq:x-wake}
\end{equation}
and if the $y$-type receiver wakes up, 
\begin{equation}
    \begin{aligned}
    X_{n+1} &= X_{n} + \textbf{1}_{x} \xi_n,\\
    Y_{n+1} &= Y_{n} -1 + \textbf{1}_{y} \xi_n.
\end{aligned}\label{eq:y-wake}
\end{equation}
 where the indicator function $\textbf{1}_{x}$ means that the receiver makes a negative comment while $\textbf{1}_{y}$ indicates the opposite (the positive comment). The total population is updated by $Z_{n+1} = Z_{n} -1 + \xi_n$.

The probability of a receiver who receives the post with a negative comment also commenting negatively can be characterized by a negative-to-negative factor $\alpha_{xx}(s)$ depending on the tag $s$. Likewise, we denote by a positive-to-negative factor $\alpha_{yx}(s)$ the probability of a receiver who receives a positive comment commenting negatively. As the receiver's comment only depends on the belief $\mu_s$ [see the best response in \eqref{eq:a-star}], $\alpha_{xx}(s)=\alpha_{yx}(s)=1-a^*(\mu_s)=1-\E_{\mu_s}[\omega]$. Simply put, the higher the $E_{\mu_s}[w]$, the more confident the receiver is about the content accuracy, and hence, the less likely the receiver is to give a negative comment.

\subsection{Stochastic Approximation Analysis}

To analyze the limit behavior of the process, we apply stochastic approximation \cite{SA_for_BP} and consider the continuous-time dynamics of the multiple-type branching. Since there are only two types in the branching process, it suffices to consider the dynamics of the total population and that of the $x$-type. Toward this end, let $\Bar{Z}_n = \frac{Z_n}{n}$, $\Bar{X}_n = \frac{X_n}{n}$, and $\gamma_n = \frac{1}{n+1}$, and then we aggregate the branching equations in \eqref{eq:x-wake} and \eqref{eq:y-wake}, leading to the following:  
\begin{equation}
    \begin{aligned}
    \Bar{Z}_{n+1} = \Bar{Z}_n &+ \gamma_{n}\big(\xi_n -1-\Bar{Z}_n\big) \textbf{1}_{\{\Bar{Z}_n > 0\}}, \\
    \Bar{X}_{n+1} = \Bar{X}_n &+ \gamma_{n}\big[\textbf{1}_{\{x-wakes\}}\big(\textbf{1}_{x} \xi_n -1\big)\\
    &+ \textbf{1}_{\{y-wakes\}}\textbf{1}_{x} \xi_n - \Bar{X}_n \big] \textbf{1}_{\{\Bar{Z}_n > 0\}},
\end{aligned}\label{eq:discrete-ode}
\end{equation}
where $\mathbb{E}[\textbf{1}_{\{x-wakes\}}]=\frac{\Bar{X_n}}{\Bar{Z_n}}$, $\mathbb{E}[\textbf{1}_{\{y-wakes\}}]=1-\frac{\Bar{X_n}}{\Bar{Z_n}}$ indicate the probabilities of an individual of $x$-type and $y$-type wakes up. Let $\Bar{X}_0 = X_0$, $\Bar{Z}_0 = X_0+Y_0$ be the initial conditions. As the discrete-time trajectory of \eqref{eq:discrete-ode} is an asymptotic pseudo-trajectory of the continuous-time system in \eqref{eq:cont-ode} \cite{SA_for_BP}, the two systems share the same limiting behavior. Hence, we arrive at \Cref{prop:sa}.
\begin{equation}
    \begin{aligned}
    \Dot{z} &= h^z(z, x) = (m-1-z)\textbf{1}_{\{z > 0\}}, \\
    \Dot{x} &= h^x(z, x) = \big[\eta\big(\alpha_{xx}(s) \cdot m-1\big)\\ &+ (1-\eta)\alpha_{yx}(s) \cdot m - x \big]\textbf{1}_{\{z > 0\}}, \eta = \frac{x}{z}
\end{aligned}\label{eq:cont-ode}
\end{equation}
\begin{proposition}
    For $\mathbb{E}[N^2] < \infty$, the $\{\Bar{Z}_{n}\}, \{\Bar{X}_{n}\}$ sequences converge to $\Bar{Z}^*, \Bar{X}^*$ almost surely, where $\Bar{Z}^* = m-1$ and $\Bar{X}^*=\eta^*(s) \Bar{Z}^*$ with $\eta^*(s) = \frac{\alpha_{yx}(s)}{1-\alpha_{xx}(s)+\alpha_{yx}(s)}$ are solutions to \eqref{eq:cont-ode}.
\label{prop:sa}
\end{proposition}

The proof for the above proposition follows \cite{BPBP_tagging}. Note that $\eta^*(s)$ and $\eta^*(a)$ can be used interchangeably because the receiver decides an action $a$ based on the posterior belief $\mu_s$ with respect to the tag $s$. Since the receiver's comment only depends on the belief, we can characterize the limiting trend under tag $s$ by the following statement.
\begin{corollary}
   As $\alpha_{yx}(s) = \alpha_{xx}(s) = 1-\E_{\mu_s}[\omega]$, then the proportion of negative comments $\eta^*(s)=\eta^*(a(\mu_s))=\alpha_{yx}(s)=1-\E_{\mu_s}[\omega]$.
\label{remark}
\end{corollary}

\subsection{Optimality Conditions under Stable Branching}
Given the receiver's best response $a^*(\mu_s)$ and the stabilized branching process, one can simplify the agent's problem, as the trend $\eta^*(s)$ admits a simple formula. Since $\eta^*(a)=1-\E_{\mu}[\omega]$ from \Cref{remark}, we notice that $v_A(\mu)=r_A(a^*(\mu))=1-\eta^*(a)=\E_{\mu}[\omega]=\mu(1)$, which is linear in $\mu(1)$. In the binary-state case, the belief $\mu_s$ is uniquely determined by its second entry $\mu(1)$. Hence, the following discussion will treat $\mu_s$ as a scalar. The same treatment also applies to the prior $p$. The agent's optimality conditions under the signaling in \eqref{eq:lambda-argmax} can be rewritten as 
\begin{align*}
    \max_{\lambda\in [0,1]}\langle p(\lambda), \Vec{v}_A(\pi)  \rangle -c(\lambda).
\end{align*}
Due to the linearity of the first term and the convexity of the second term, the optimality follows the first-order condition: 
\begin{align}
    \langle p_1-p_0, \Vec{v}_A(\pi) \rangle =\nabla c(\lambda), \label{eq:agent-opt}
\end{align} As later shown in the ensuing section, the agent's marginal cost $\nabla c$ plays a significant part in the feasibility of the sender's information structures. 

Since $\eta^*(a)=1-\E_{\mu}[\omega]$, the sender's expected utility under the belief $\mu$  is $v_{S}(\mu)=-(1-\E_{\mu}[\omega])\E_{\mu}[\omega]+\E^2_{\mu}[\omega]$. In the binary-state case, $v_S(\mu)=\mu(\mu-1)+\mu^2$. Hence, the sender's problem is given by 
\begin{equation}
    \begin{aligned}
    \max_{\pi, \lambda} & \langle p(\lambda), \Vec{v}_S(\pi) \rangle \\
    \text{s.t. } & \langle p_1-p_0, \Vec{v}_A(\pi) \rangle =\nabla c(\lambda) \\
    & \text{consistent belief system in } \eqref{eq:consistency}
\end{aligned}\label{eq:sender-old}
\end{equation}
Note that the agent's decision variable $\lambda$ also appears in the maximization, as we assume that the tie breaks in favor of the sender should there exists multiple effort level $\lambda$ satisfying the first constraint in \eqref{eq:sender-old}.  The consistency requirement in \eqref{eq:consistency} involves division operation, leading to a highly nonlinear programming problem. To simplify our analysis,  the proposition in the following section \ref{sec:plausibility} transforms the sender's problem into the posterior belief space using Bayesian Plausibility.

\section{Perfect Bayesian Equilibrium Characterization: A Lagrangian Approach}
\subsection{Bayesian Plausibility} \label{sec:plausibility}

 Bayesian plausibility is a sanity check for any information structure: all possible posterior beliefs induced by the realized signals should be consistent with the prior under the information structure. Formally,  the proposition below reformulates the sender's problem where the decision variable is a distribution over posteriors $\tau\in \Delta(\Delta(\Omega))$. 
\begin{proposition}[Bayesian Plausibility]
    Given an effort $\lambda$, there exists a signaling mechanism (tagging policy) $\pi$ satisfying the conditions in \eqref{eq:sender-old} if and only if there exists a distribution over posteriors $\tau\in \Delta(\Delta(\Omega))$ such that 
    \begin{align*}
    \E_{\tau}[\mu] = p(\lambda),\E_{\tau} \left[\E_{\mu}[\nabla \log p(\lambda) ]v_A(\mu)\right]=\nabla c(\lambda).
    \end{align*}
\end{proposition}
\begin{proof}
    We first prove the equivalence between the signaling mechanism $\pi$ and the distribution $\tau$. Without loss of generality, assume that for each signal $s\in\Sigma$, the receiver has a distinct posterior belief $\mu_s$. Starting from $\pi$, and fixing $\lambda$, the probability of generating $\mu_s$ is 
    \begin{equation*}
        \tau(\mu_s)=\sum_{\omega} \pi(s|\omega) p(\omega;\lambda)=\langle \pi(s|\cdot), p(\lambda) \rangle.
    \end{equation*}
    Following the above equation, one can compute the distribution of posteriors using the signaling. Conversely, recall that the Bayes rule gives 
    \begin{equation*}
        \mu_s=\frac{\pi(s|\cdot)\odot p(\lambda)}{\langle \pi(s|\cdot), p(\lambda)\rangle }= \frac{\pi(s|\cdot)\odot p(\lambda)}{\tau(\mu_s)},
    \end{equation*}
    implying that $\pi(s|\cdot)=\tau(\mu_s)(\mu_s\oslash p(\lambda))$, where $\oslash$ denotes the point-wise division. The equation above indicates that one can recover the signaling using the distribution of posteriors $\tau$. Finally, note that 
    \begin{align*}
         \pi(s|\cdot)=\tau(\mu_s)(\mu_s\oslash p(\lambda))\Leftrightarrow \sum_{s}\pi(s|\cdot) p(\lambda)=\sum_{s}\tau(\mu_s)\mu_s,
    \end{align*}
    which proves the first equality in the proposition. The posterior distribution $\tau$ associated with $\pi$ is called the Bayesian-plausible distribution in the literature \cite{kamenica11BP}.  

    To recover the agent's optimality condition (also called incentive-compatibility constraint), consider the constraint:
    \begin{align*}
        \langle p_1 - p_0,  \Vec{v}_A(\pi) \rangle =\nabla c(\lambda). 
    \end{align*}
    Plugging the above equation into the left-hand side gives 
    \begin{align*}
        &\langle p_1- p_0,   \Vec{v}_A(\pi) \rangle \\
        &= \sum_{\omega}\left(\sum_{s}\pi(s|\omega)v_A(\mu_s)\right)(p_1(\omega)-p_0(\omega))\\
        &=\sum_{\omega} \left(\sum_{s}\tau(\mu_s)\frac{\mu_s(\omega)}{\pi(\omega;\lambda)}v_A(\mu_s)\right)(p_1(\omega)-p_0(\omega))\\
        &=\sum_{s}\tau(\mu_s)\sum_{\omega} \left(\frac{p_1(\omega)-p_0(\omega)}{p(\omega;\lambda)}\mu_s(\omega)\right)v_A(\mu_s)\\
        &=\sum_{s}\tau(\mu_s)\sum_{\omega} \left(\frac{\nabla_\lambda p(\omega;\lambda)}{p(\omega;\lambda)}\mu_s(\omega)\right)v_A(\mu_s)\\
        &=\E_{\tau}[\E_{\mu}[\nabla_\lambda \log p(\omega;\lambda)] v_A(\mu)]
    \end{align*}
\end{proof}

Let $f(\mu)=\E_{\mu}[\nabla_\lambda \log p(\omega;\lambda)] v_A(\mu)-\nabla c(\lambda)$, then the sender's problem can be rewritten as 
\begin{align}
\max_{\tau\in\Delta(\Delta(\Omega)), \lambda} & \E_{\tau}[v_S(\mu)],\label{eq;sender-max}\\
 \text{s.t. } & \E_{\tau}[\mu]=p(\lambda),\label{eq:bp}\\
& \E_{\tau} [f(\mu)]=0,\label{eq: ic}
\end{align}
where \eqref{eq:bp}, referred to as the Bayesian plausibility constraint (BP), corresponds to the consistency in \eqref{eq:consistency}; \eqref{eq: ic}, referred to as the incentive-compatibility constraint (IC), rephrases the agent's optimality condition in \eqref{eq:agent-opt}. 

\subsection{The Lagrangian Characterization}

With Bayesian plausibility, the sender's problem becomes equality-constrained nonlinear programming, which naturally prompts one to consider the Lagrange multiplier method. In what follows, we present a PBE characterization through the lens of Lagrangian. The discussion begins with the feasible domain of the maximization in \eqref{eq;sender-max}.  
\begin{proposition}[Implementable Effort, Feasible Condition]
In the binary-state model, let $\bar{\lambda}$ be the value such that $\nabla c(\lambda)=p_1-p_0$. Then, $\lambda$ is feasible if and only if $\lambda\leq\bar{\lambda}$.
\label{prop:feasible}
\end{proposition}
\vspace{-2mm}
\begin{proof}
    We begin with the necessity. In the binary-state model, the incentive compatibility (IC) constraint reduces to 
    \begin{align*}
        (p_1-p_0)(\bar{v}_A(1)-\bar{v}_A(0))=\nabla c(\lambda),
    \end{align*}
    where $\bar{v}_A(\omega|\pi)=\sum_{s}\pi(s|\omega)v_A(\mu_s)$. Note that $v_A(\mu)=\mu\in [0,1]$, implying that $\bar{v}_A$ never exceeds $1$, and so is the difference $\bar{v}_A(1)-\bar{v}_A(0)$. Hence, $p_1-p_0\geq \nabla c(\lambda)$. As the cost function $c$ is strictly increasing, $\nabla c(\lambda)> \nabla c(\bar{\lambda})=p_1-p_0$, for $\lambda>\bar{\lambda}$, which means that $\lambda$ is not IC. 

    For sufficiency, consider $\lambda\in [0, \bar{\lambda}]$, and $p(\lambda)=(1-\lambda, \lambda)$. Let $\Delta p=p_1-p_0$, we construct a Bayesian-plausible $\tau$ as follows and refer it as a ``hybrid tagging policy''. $\operatorname{supp}(\tau)=\{0, \lambda, 1\}$ (these scalars denote the second entries of beliefs), and 
    \begin{align*}
        \tau(0)=\frac{(1-\lambda)\nabla c(\lambda)}{\Delta p}, \tau(\lambda)=1-\frac{\nabla c(\lambda)}{\Delta p}, \tau(1)=\frac{\lambda \nabla c(\lambda)}{\Delta p}. 
    \end{align*}
    It is straightforward to verify that this posterior distribution satisfies both constraints in the sender's problem. This construction implies that for any $\lambda\in [0, \bar{\lambda}]$, one can find a feasible $\tau$, and hence, $\lambda$ is also implementable.
\end{proof}
\begin{remark}
The maximum effort the receiver is willing to exert $\bar{\lambda}$ in uncovering the truth solely depends on their marginal cost $\nabla c(\bar{\lambda})$, regardless of their reputation. The higher the marginal cost $\nabla c(\lambda)$ is, the smaller the upper bound $\bar{\lambda}$ is, leading to a modest feasible set, which means that the agent cannot afford to create authentic content in this case, regardless of their reputation.
\end{remark}

The term ``hybrid'' of the constructed $\tau$ in the proof above is due to the observation that $\tau$ is a convex combination of two representative tagging policies. Consider $\overline{\tau}$ and $\underline{\tau}$: $\supp(\overline{\tau})=\{0, 1\}$, $\overline{\tau}(0)=1-\lambda$, and $\overline{\tau}=\lambda$; $\supp(\underline{\tau})=\{\lambda\}$ and $\overline{\tau}(\lambda)=1$. $\overline{\tau}$ is the fully informative tagging, where the receiver, upon receiving the tag, is certain about the accuracy: the post is either fake $0$ or true $1$. In contrast, $\underline{\tau}$ is the opposite: the uninformative tagging. The corresponding belief system is degenerate, including only one belief that is exactly the prior distribution. This degeneracy indicates that the receiver does not acquire any helpful information from the tag to update the belief on the content accuracy.         

From the above construction, we arrive at the following proposition, stating that the sender strictly prefers and incentivizes the agent to exert a positive effort level. 
\begin{proposition}[Positive Effort Level]
$\lambda=0$ is implementable under the uninformative signaling: $\pi(\cdot|\omega)=\operatorname{unif}(\Sigma)$ for any $\omega\in \Omega$. This uninformative signaling is strictly dominated by the hybrid signaling with $\lambda\in (0,\bar{\lambda})$. 
\label{cor:Positive effort}
\end{proposition}
\begin{proof}
    From \Cref{assump:cost}, $\lambda=0$ implies $\nabla c(\lambda)=0$, and further implies that 
    \begin{align*}
        (p_1-p_0)\left(\sum_{s}\pi(s|1)v_A(\mu_s)-\sum_{s}\pi(s|0)v_A(\mu_s)\right)=0
    \end{align*}
    The uninformative signaling naturally satisfies the above equation; hence, $\lambda=0$ is implementable. This uninformative signaling ($\lambda=0$) induces a degenerate posterior distribution: $\operatorname{supp}(\tau)=\{p_0\}$, and the sender's expected utility is $0$. In contrast, consider the signaling in the proof above, $\operatorname{supp}(\tau)=\{0, \lambda, 1\}$, $\lambda\in (0, \bar{\lambda})$, with $\tau(0)$, $\tau(\lambda)$, and $\tau(1)$ as in \Cref{prop:feasible}.
    Note that $v_S(1)=1$, the sender's expected utility is $\E_\tau[v_S(\mu)]=\lambda>0$.
\end{proof}

\begin{corollary}
    As long as the set $(0, \bar{\lambda})$ is not empty, the sender can always create a tagging policy that incentivizes the agent to invest positive effort in discovering the truth, regardless of the value of the cost function.
\end{corollary}

The above discussion addresses the agent's feasibility condition. In what follows, we shift the focus to the sender's problem, given an implementable effort $\lambda$. Denote by $\tau^\lambda$ and $V^\lambda$ the optimal solution to the sender's problem \eqref{eq;sender-max} (fixing $\lambda$), and the corresponding objective value, respectively. Consider the following set $F^\lambda\subset \R^{|\Omega|+2}$: $F^\lambda=\{(\mu, f(\mu), v_S(\mu)): \mu\in \Delta(\Omega)\}$. By construction, each entry of any element in $F^\lambda$ corresponds to the integrand in the three objects in the sender's problem \eqref{eq;sender-max}. These integrands are referred to as ex-post values. Denote by $co(F^\lambda)$ the convex hull of $F^\lambda$, including all the ex-ante values  that can be generated using a probability $\tau\in \Delta(\Delta(\Omega))$. A standard argument from constrained programming gives the following. 
\begin{proposition}
    Given an implementable effort $\lambda$, the maximal utility the sender can attain is $V^\lambda=\max\{v: (p(\lambda), 0, v)\in co(F^\lambda)\}$.
\end{proposition}
\begin{proof}
    It suffices to note that $\mu=p(\lambda)$ naturally satisfies \eqref{eq:bp}, and $f(\mu)=0$ induces \eqref{eq: ic}. Therefore, any point $(\mu, f(\mu), v)\in \{(p(\lambda), 0, v)\in co(F^\lambda)\}$ is feasible to \eqref{eq;sender-max}. Therefore, $V^\lambda$, as a convex combination of these points, is the maximal value.   
\end{proof}

The above proposition gives a geometric intuition where the solution should be: $(p(\lambda), 0, V^\lambda))$ lies on the boundary of the convex set $co(F^\lambda)$. Hence, there exists a supporting hyperplane at $(p(\lambda), 0, V^\lambda))$, leading to the following.
\begin{proposition}[Lagrangian Characterization]
    Given an implementable $\lambda$, a distribution of posteriors $\tau^\lambda$ is a solution to the sender's problem if and only if it satisfies \eqref{eq:bp}, \eqref{eq: ic}, and there exists $\psi\in \R$, $\rho\in \R$, and $\varphi\in \R^{|\Omega|}$ such that 
    \begin{equation*}
        \mathcal{L}(\mu, \psi,\varphi)=v_S(\mu)+\psi f(\mu)-\langle \varphi , \mu \rangle \leq \rho , \text{ for all } \mu\in \Delta(\Omega),
    \end{equation*}
    where the equality holds for all $\mu$ such that $\tau^\lambda(\mu)>0$.
\label{prop:Lagrangian}
\end{proposition}
\begin{proof}
We begin with the necessity. As $(p(\lambda),0, V^\lambda )$ is a boundary point of a closed convex set, the separating hyperplane theorem tells that there exists a normal vector $d=(-\varphi, \psi, 1)\in \R^{|\Omega|+2}$ and a scalar $\rho$ such that $\langle d, y\rangle \leq \rho$ for all $y\in co(F^\lambda)$, where the equality holds for $y=(p(\lambda), 0, V^\lambda)$. Rearranging terms in this inner product, we obtain that $ \mathcal{L}(\mu, \psi,\varphi)\leq \rho $. 

It remains to show that $\mathcal{L}(\mu, \psi,\varphi)=\rho$ for all $\mu\in \{\mu:\tau^\lambda(\mu)>0\}$. Suppose, for the sake of contradiction, that there exists some $\mu\in \operatorname{supp}(\tau^\lambda)$ such that  $\mathcal{L}(\mu, \psi, \varphi)<\rho$. Note that $\mathcal{L}(\mu, \psi, \varphi)\leq\rho$, then $V^\lambda=\E_{\tau^\lambda}[\mathcal{L}(\mu, \psi, \varphi)]< \rho$. Rearranging terms, we obtain $\langle d, (p(\lambda), 0, V^\lambda)\rangle < \rho$, which contradicts the fact that the supporting hyperplane passes through the point $(p(\lambda), 0, V^\lambda)$.

For the sufficiency part, if $v_S(\mu)+\psi f(\mu)\leq \rho + \langle \varphi, \mu\rangle $ for all $\mu \in \Delta(\Omega)$, then for any $\tau$, 
\begin{equation*}
    \E_{\tau }[v_S(\mu)]+\psi \E_{\tau}[f(\mu)]\leq \rho+\E_\tau[\langle \varphi, \mu\rangle].
\end{equation*}
Since $\tau^\lambda$ satisfies \eqref{eq:bp} and \eqref{eq: ic}, the above reduces to $\E_{\tau^\lambda }[v_S(\mu)]\leq \rho+ \langle \varphi, p(\lambda)\rangle$. If $\tau^\lambda$ is such that $\mathcal{L}(\mu,\psi, \varphi)=\rho$, for all $\mu\in\operatorname{supp}(\tau^\lambda)$, then $\E_{\tau^\lambda}[v_S(\mu)]=\rho+\langle \varphi, p(\lambda)\rangle $, meaning that the expected utility $\E_\tau[v_S(\mu)]$ reaches the upper bound at $\tau^\lambda$. 
\end{proof}

Fixing $\lambda\in (0, \bar{\lambda})$, consider the Lagrangian function $\mathcal{L}$ introduced in the above. Its second-order derivative is given by $\frac{\partial^2 \mathcal{L}}{\partial \mu^2} = \nabla^2 v_S(\mu) +\frac{2\psi}{\lambda(1-\lambda)}$. From the definition, $\nabla^2 v_S(\mu)>0$, and hence, the sign of $\frac{\partial^2 \mathcal{L}}{\partial \mu^2}$ depends on $\psi$, for which we have the following characterization. 
\begin{proposition}
    For any $\lambda\in (0,\bar{\lambda}]$, the Lagrange multiplier $\psi$ associated with the solution $\tau^\lambda$ is non-positive. 
\end{proposition}
\begin{proof}
Consider a relaxation to the original problem without IC constraint \eqref{eq: ic}:
\begin{align}
    \widetilde{V}^\lambda=\max_{\tau\in \Delta(\Delta(\Omega))} \E_{\tau}[v_S(\mu)] \text{ subject to } \E_\tau[\mu]=p(\lambda),
\end{align}
which is exactly the standard Bayesian persuasion \cite{kamenica11BP}.

Denote by $\tilde{\tau}^\lambda$ the solution to the relaxed problem when fixing $\lambda$. Applying the Lagrangian characterization developed in Proposition 6, there exists $\tilde{\rho}$ and $\tilde{\varphi}$ such that $v_S(\mu)\leq \tilde{\rho} + \tilde{\varphi} \mu$, for all $\mu\in [0,1]$, with equality if $\tilde{\tau}(\mu)>0$.  Define $g(\lambda)=\E_{\tilde{\tau}}[f(\mu)]$. Let $\tau^\lambda$ be the solution to the original problem. We aim to prove $\psi g(\lambda)\leq 0$ in the following. The definition of two Lagrangians give
\begin{align}
    \rho+\varphi \lambda=\E_{\tau^\lambda}[v_S(\mu)]\leq \E_{\tilde{\tau}^\lambda}[v_S(\mu)]=\tilde{\rho}+\tilde{\varphi}\lambda.\label{eq:tilde-ineq}
\end{align}
Finally, taking the expectation of the original Lagrangian in \Cref{prop:Lagrangian} with respect to $\tilde{\tau}$, we obtain 
\begin{equation}
    \E_{\tilde{\tau}}[v_S(\mu)]+\psi \E_{\tilde{\tau}}[f(\mu)]\leq \rho +\varphi \lambda \Leftrightarrow  \tilde{\rho} +\tilde{\varphi} \lambda +\psi g(\lambda) \leq \rho +\varphi \lambda \label{eq:expect-ineq}
\end{equation}
Combining \eqref{eq:expect-ineq} and \eqref{eq:tilde-ineq} leads to $\psi g(\lambda)\leq 0$. 

The rest of the proof establishes that $g(\lambda)\geq0$ for $\lambda \in (0,\bar{\lambda}]$. Note that the sender's expected utility $v_S(\mu)=2\mu^2-\mu$ is convex in $\mu$. The standard persuasion analysis gives that the unique optimal signaling is the fully informative one \cite[Section 3]{kamenica11BP}, implying that $\operatorname{supp}(\tilde{\tau})=\{0,1\}$, and $\tilde{\tau}(0)=1-\lambda$, $\tilde{\tau}(1)=\lambda$. Direct calculation yields $f(\mu)=\frac{\mu(\mu-\lambda)}{\lambda(1-\lambda)}-\nabla c(\lambda).$ Hence, $g(\lambda)=\E_{\tilde{\tau}}[f(\mu)]=1-\nabla c(\lambda)\geq 0$, for $\lambda\in (0,\bar{\lambda}]$, implying that $\psi\leq 0$.

\end{proof}
Even though it seems that $\frac{\partial^2 \mathcal{L}}{\partial \mu^2}$ is not necessarily non-negative, the following proposition asserts that the Lagrangian must be a convex function of $\mu$, which leads to the main conclusion of this work: the sender's optimal signaling is the fully informative one, under which the agent is incentivized not to create misinformation to the best effort.     
\begin{proposition}
\label{prop:optimal-is}
    The Lagrangian function is convex with respect to $\mu$, and hence, the optimal signaling is fully informative and implements $\bar{\lambda}$. 
\end{proposition}
\begin{proof}
    Suppose, for the sake of contradiction, that the multiplier $\psi$ associated with the solution is such that $\frac{\partial^2 \mathcal{L}}{\partial \mu^2}<0$. This contradiction indicates the convexity of the Lagrangian. Notice that $\nabla^2 v_S(\mu)=4$, and that $\frac{\partial^2 \mathcal{L}}{\partial \mu^2}$ is a constant, then one can see that the Lagrangian function is strictly concave everywhere. Therefore, the sender's optimal signaling is degenerate (only one belief) and is strictly dominated (see \Cref{cor:Positive effort}), which contradicts the optimality. Consequently, the standard argument from Bayesian persuasion literature \cite[Section 3]{kamenica11BP} also applies to the proposed model, leading to the statement above.  
\end{proof}

\begin{corollary}
   A viable ``prevent than cure'' solution to misinformation is simply the most straightforward tagging policy: revealing the truth to the user. 
    \label{cor:incentive provision}
\end{corollary}

Lastly, we discuss the impact of the cost function $c(\lambda)$  on the agent's reputation  $1-\eta^*$, from which we further elaborate on how the sender provides the agent incentive to spend effort. Assume the cost function $c(\lambda)$ takes the quadratic form: $c(\lambda) = k\lambda^2$, and $k>\frac{1}{2}$ so that $\nabla c(1)>1$. 
\begin{proposition}
\label{prop:cost}
    Under the hybrid tagging, the equilibrium trend $\E_{\tau}[\eta^*(\mu)]$ admits the following characterizations, depending on the Hessian parameter $k$: 1) for $k\geq 1$, $\bar{\lambda}\leq 1/2$. $\E_{\tau}[\eta^*(\mu)]=1-\lambda\geq 1/2$, for $\lambda\in (0,\bar{\lambda}]$; 2) for $1/2< k < 1$, $\bar{\lambda}>1/2$, $\E_{\tau}[\eta^*(\mu)]=1-\lambda<1/2$, for $\lambda\in (1/2,\bar{\lambda})$. 
\end{proposition}
\begin{remark}[Indirect Incentive Provision]
    $k=1, \bar{\lambda}=1/2$ is a turning point of the equilibrium trend. When it is costly for the agent to produce trustworthy content ($k\geq 1$), the average comment turns against themselves. The best the agent can hope is to exert the highest effort $\bar{\lambda}$ and to keep the user in the neutral position ($\E[\eta^*(\mu)]=1-\bar{\lambda}=1/2$). In contrast, the agent can earn a reputation for being a reliable information source when $\bar{\lambda}>1/2$, as the equilibrium trend under the maximum effort is in favor of themselves $\E[\eta^*]=1-\bar{\lambda}<1/2$. To sum up, the agent is always willing to exert the highest effort, whatever the cost is; and the sender achieves such an incentive provision through the receiver's action.  
\end{remark}

\section{Numerical Studies}\label{sec:numerical}
This section studies the proposed Bayesian persuaded branching processes model in three different tagging policies: fully informative, uninformative, and hybrid informative tagging policy. For each experiment, the branching setup is given by  $X_0=Y_0=50, m_{N}=50, q=0.5$, and $\tau_{n}=\tau_{1500}$. The numerical results in this section are the average of  $1000$ independent simulations.
\begin{figure*}
    \centering
\begin{subfigure}[t]{0.32\textwidth}
    \includegraphics[width=\textwidth]{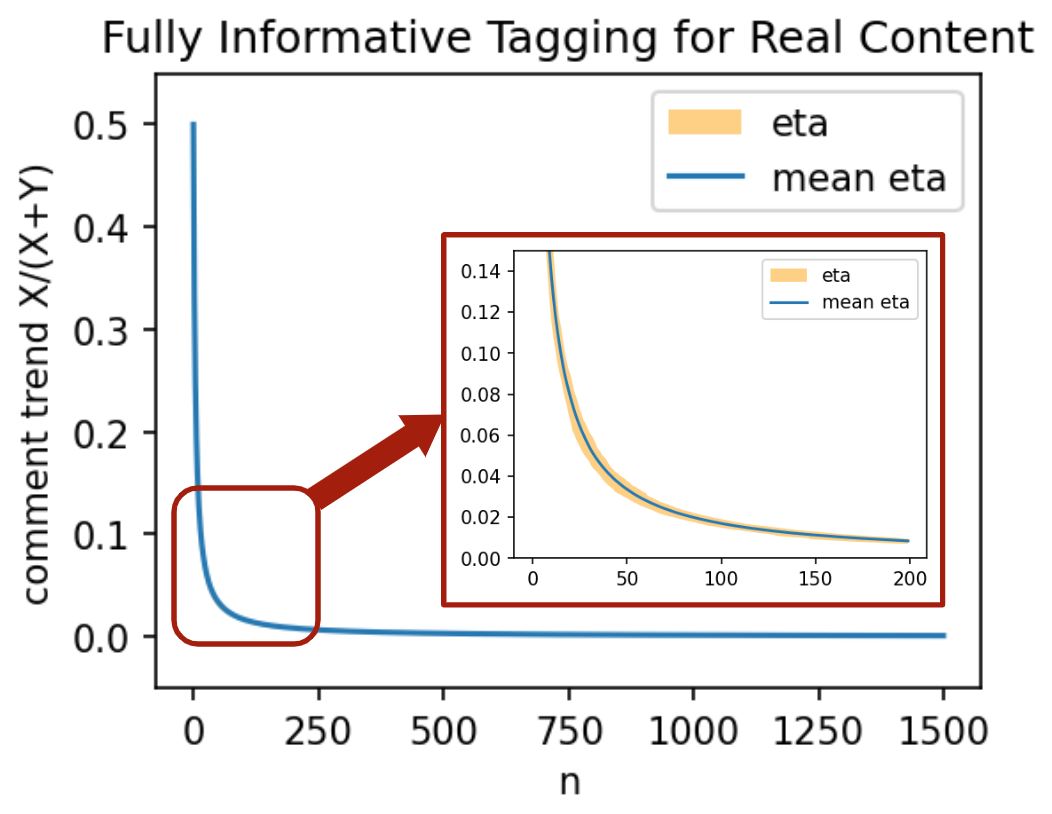}
    \caption{The authentic post yields a positive trend under the fully informative policy: $\eta=0$ when $\omega=1$.}
    \label{fig:fully_real}
\end{subfigure}
\hfill
\begin{subfigure}[t]{0.32\textwidth}
    \includegraphics[width=\textwidth]{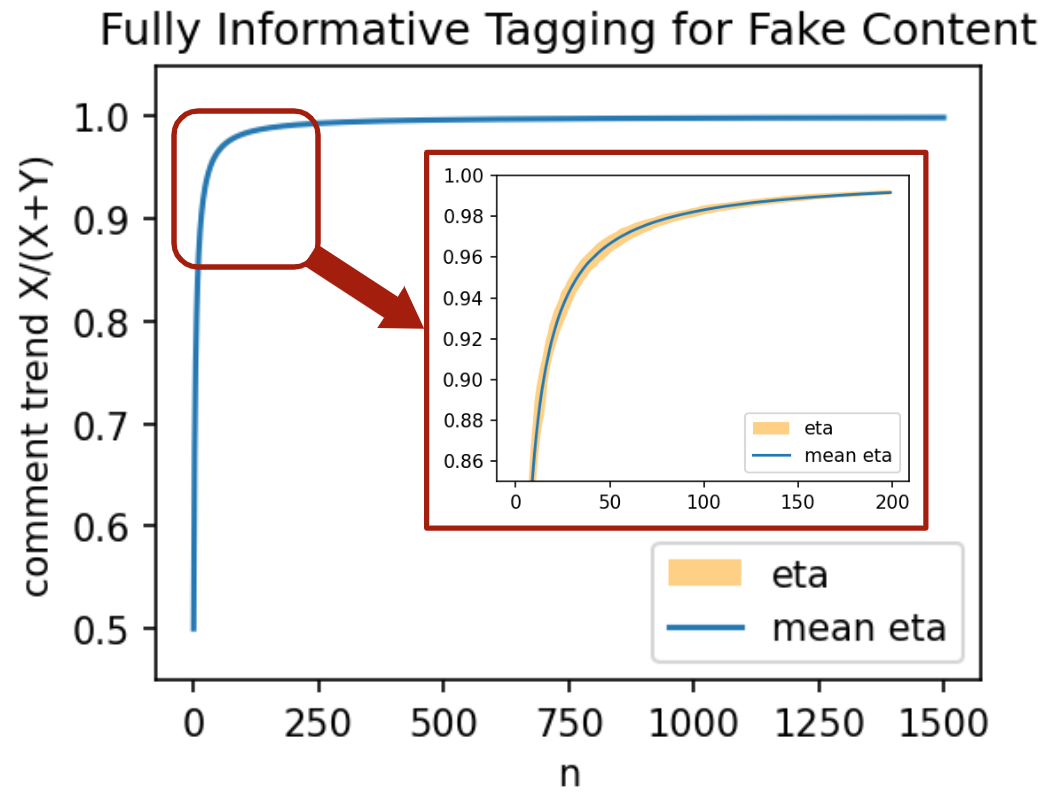}
    \caption{The post with misinformation yields a negative trend under the Fully informative/uninformative policy: $\eta=1$ when $\omega=0$.}
    \label{fig:fully_fake}
\end{subfigure}
\hfill
\begin{subfigure}[t]{0.32\textwidth}
    \includegraphics[width=\textwidth]{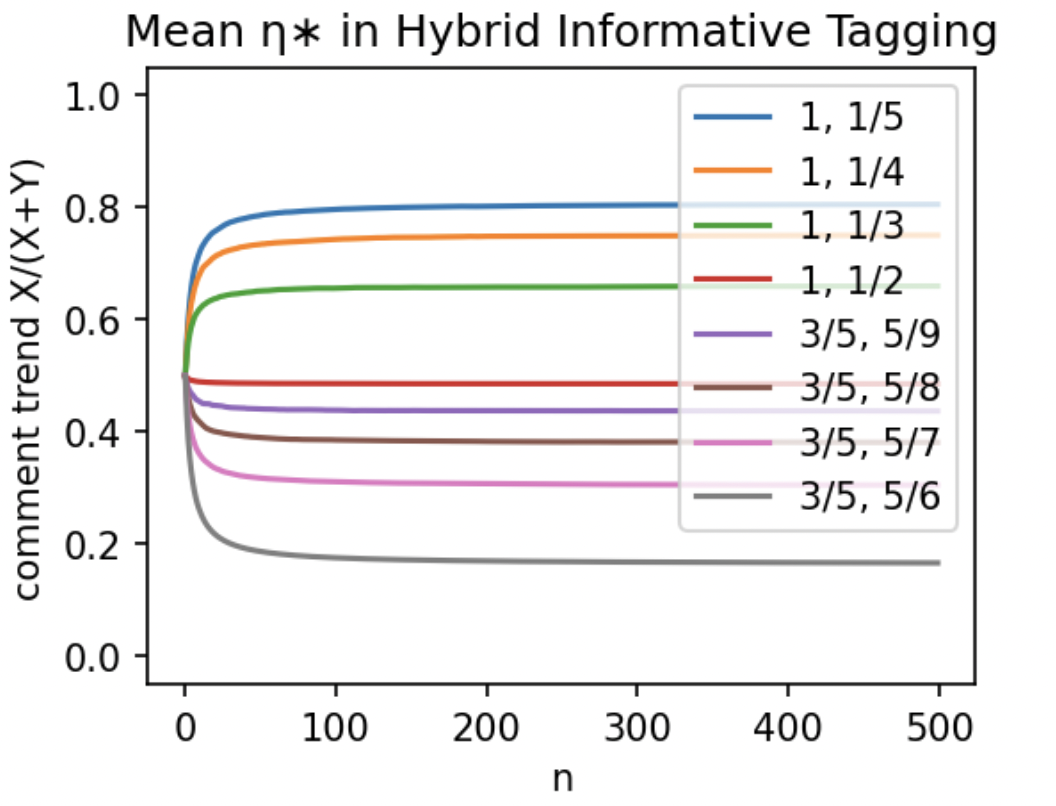}
    \caption{Trends under hybrid policies with different $k$ (former in the label) in the cost function and effort $\lambda$ (latter in the label).}
    \label{fig:hybrid}
\end{subfigure}
\vspace{-2mm}\caption{Simulations of online misinformation circulation under three representative tagging policies.}\vspace{-5mm}
\end{figure*}

\paragraph{Fully Informative Tagging Policy}
In this scenario, the policy for the OSP is to tag the post with its true state, i.e., $s= \omega$. For the authentic post, $\omega=1, s=1, \E_{\mu_s}[\omega]=1$,
and thus $\alpha_{yx}=\alpha_{xx}= 1-\E_{\mu_s}[\omega]=0$. The result for the proportion of negative comments $\eta^*$ is shown in \Cref{fig:fully_real}. On the other hand, for misinformation post, $\omega=0, s=0, \E_{\mu_s}[\omega]=0$
and thus $\alpha_{yx}=\alpha_{xx}= 1-\E_{\mu_s}[\omega]=1$. Under this tagging policy, each tag carries no ambiguity; consequently, the receiver is certain about the content's accuracy and comment on the post accordingly. As a result, the fully informative tagging leads to a positive trend for the authentic post [see \Cref{fig:fully_real}], while a negative one for misinformation [see \Cref{fig:fully_fake}].  The shaded yellow region in the figure indicates the standard deviation of $\eta^*$, while the blue line represents the mean.

\paragraph{Uninformative Tagging Policy}
Under the uninformative tagging, the OSP tags the post randomly, i.e., choosing $s=0$ and $s=1$ with probability $\frac{1}{2}$ regardless of $\omega$. According to \Cref{cor:Positive effort}, $\lambda=0$ and $\E_{\mu_s}[\omega]=0$, which leads to  $\alpha_{yx}=\alpha_{xx}= 1-\E_{\mu_s}[\omega]=1$ and $\eta^*=1$. The trend evolution is the same as in \Cref{fig:fully_fake}.

\paragraph{Hybrid Tagging Policy}  
We finally consider the hybrid tagging policy in \Cref{prop:feasible}. The cost function is of the quadratic form as in \Cref{prop:cost}. For $k=1$, the maximum feasible effort is $\bar{\lambda}=\frac{1}{2}$, under which the trend is neural: $\eta^*=0.5$. For any other $\lambda \in (0, \bar{\lambda})$, however, the resulting $\eta^*$ is strictly greater than half, as demonstrated in the upper side of \Cref{fig:hybrid}. The numerical results coincide with the analysis in \Cref{prop:cost}, showing that the agent needs to exert the best effort to investigate the truth so as not to hurt their reputation.  For $\frac{1}{2} < k < 1$, we take $k=\frac{3}{5}$ as an example. In this case, the maximum effort is $\bar{\lambda}=\frac{5}{6}$, and any implementable $\lambda > \frac{1}{2}$ leads to positive trends as shown in the lower part of \Cref{fig:hybrid}. The more effort the agent spends, the more positive the trend is; hence, the higher reputation the agent earns.

\section{Conclusion}
This work has investigated a preemptive approach to mitigate misinformation spread on OSP by disincentivizing the content creator to create misleading content in the first place. We have developed a three-player persuasion game to model the strategic interaction among the OSP, the content creator, and the user. By transforming the perfect Bayesian equilibrium into the posterior belief space, we have reformulated the OSP's equilibrium problem as an equality-constrained nonlinear programming (with a convex objective), which admits a concise Lagrangian characterization. The convexity of the Lagrangian implies that the OSP can solicit the best effort from the content creator in reducing misinformation, even though the OSP exerts no direct control over the content creator. One direction of the future work would be to investigate other mitigation mechanisms, including verification of the accuracy of the content and the accountability of the content creators.

\bibliographystyle{ieeetr}
\bibliography{ref}

\end{document}